\documentclass[11pt,a4paper]{article}
\usepackage{graphicx}
\usepackage[hyperref]{eacl2021}
\usepackage{setspace}
\usepackage{times}
\usepackage[T1]{fontenc}
\usepackage{latexsym}

\usepackage{amsmath}
\usepackage{amsthm}
\usepackage{amsfonts}
\usepackage{xspace}
\usepackage{color}
\usepackage{subcaption}
\usepackage{array}
\newcolumntype{L}{>{\centering\arraybackslash}m{3cm}}

\newcommand{\stos}{auto-encoder\xspace}

\newtheorem{lemma}{Lemma}
\newtheorem{theorem}{Theorem}

\usepackage{microtype}

\aclfinalcopy 

\title{ADePT: Auto-encoder based Differentially Private Text Transformation}

\author{Satyapriya Krishna \\
  Amazon Alexa \\
  \texttt{satyapk@amazon.com} \\\And
  Rahul Gupta  \\
  Amazon Alexa \\
  \texttt{gupra@amazon.com} \\\And
  Christophe Dupuy  \\
  Amazon Alexa \\
  \texttt{dupuychr@amazon.com} \\}

\begin{document}

\maketitle
\begin{abstract}
Privacy is an important concern when building statistical models on data containing personal information. Differential privacy offers a strong definition of privacy and can be used to solve several privacy concerns \cite{dwork2014algorithmic}. 
Multiple solutions have been proposed for the differentially-private transformation of datasets containing sensitive information. 
However, such transformation algorithms offer poor utility in Natural Language Processing (NLP) tasks due to noise added in the process. 
In this paper, we address this issue by providing a utility-preserving differentially private text transformation algorithm using auto-encoders. Our algorithm transforms text to offer robustness against attacks and produces transformations with high semantic quality that perform well on downstream NLP tasks. 
We prove the theoretical privacy guarantee of our algorithm and assess its privacy leakage under Membership Inference Attacks (MIA) \cite{shokri2017membership} on models trained with transformed data. Our results show that the proposed model performs better against MIA attacks while offering lower to no degradation in the utility of the underlying transformation process compared to existing baselines.
\end{abstract}

\section{Introduction}
Differentially Private (DP) mechanisms provide robustness against privacy attacks and offer practical solutions for transforming and releasing datasets without compromising privacy \cite{dwork2009complexity}. 
A typical downstream task may involve training a machine learning model with data transformed from a differentially private mechanism. 
However, while the DP mechanism offers privacy, it can adversely impact the utility of the trained model \cite{li2009tradeoff}.
Specifically, in the case of text datasets (e.g., those used in Natural Language Understanding(NLU) tasks), if the DP transformation impacts the syntactic structure of the sentence or does not factor in the target NLU label (e.g. intent of the sentence in an intent classification task), the loss in utility can render the use of processed data impractical.
We address this problem in the paper and introduce ADePT - an Auto encoder based Differentially Private Text transformation mechanism that process text data while reducing the impact on the utility of the dataset.
 
The ADePT mechanism relies on text-based auto-encoders (e.g. LSTM based sequence-to-sequence models) for text transformation.
An auto-encoder first transforms a given text input into some latent representation, followed by text generation (transformation) via the decoder. 
In this paper, we prove that the application of clipping and noising operation on the latent sentence representations returned by the encoder followed by text generation by the decoder is a DP mechanism. 
We use ADePT to transform text datasets relevant to the Intent Classification (IC) task, where we predict intent of input sentence (e.g. `BuyTicketIntent' intent prediction for the sentence `buy me a ticket to Seattle' ).
While one can transform the text in datasets and retain original intent labels to train the intent classifier, it is not guaranteed that the transformed text would correspond to the original intent post transformation, which can adversely impact the trained IC's utility. To mitigate this problem, we append the intent labels to the rest of tokens while training the auto-encoder as well as for transforming text with the trained auto-encoder. For instance, {\it @BuyTicketIntent buy me a ticket to Seattle} is used as the input sample to train the autoencoder where { \it @BuyTicketIntent} is the intent annotation for { \it buy me a ticket to Seattle}. Similarly, the intent label is regenerated along with the rest of the tokens after transformation, which is then used for IC training with the regenerated intent as the label of the regenerated sentence.
In addition to this, we argue that data regeneration via decoder maintains the syntactic structure of the sentence since the decoder generates tokens auto-regressively, factoring in the previously generated tokens.
We hypothesize that these properties make ADePT a utility preserving DP mechanism and demonstrate the superiority of the algorithm against an existing baseline \cite{feyisetan2019leveraging}. 

\section{Related Work}
{\bf Differentially private data transformation and generation:}
Researchers have proposed several methods for DP data transformation using individual ranking micro-aggregation \cite{sanchez2016utility}, random projection \cite{xu2017dppro}, and kernel mean embeddings \cite{balog2018differentially}. 
Alternatively, models such as differentially private Generative Adversarial Networks \cite{xie2018differentially} and differentially private autoencoder-based generative model \cite{chen2018differentially} focus on training data generators that guarantee that the data generation mechanism is DP.
While DP data transformation and generation has shown great success for structured data (e.g. numeric tables, histograms), the same for unstructured data (e.g. text) is more challenging.
 \citet{beigi2019not} propose an algorithm that learns numeric text representations that offer guarantees of differential privacy.
However, arguably it may be more desirable to release transformed text as opposed to latent representations. \citet{feyisetan2019leveraging} proposes DP mechanism to transform text data that constructs a hierarchical representation given a sentence to identify {\textit{private phrases}} in the input sentence. 
Each word in the private phrase is then randomly replaced by neighboring word in a word embedding space. 
We use the work by \citet{feyisetan2019leveraging} as baseline in our work since it also focusses on obtaining text transformed text with a DP mechanism. 

{\bf Membership Inference Attacks:}
While the $(\epsilon, \delta)$ bounds provide theoretical quantification of a mechanism's privacy \cite{dwork2014algorithmic}, recently Membership Inference Attack (MIA) success rates have emerged as practical quantification of privacy preservation \cite{shokri2017membership}.
In this work we use the setup suggested by \citet{shokri2017membership} as a method to quantify privacy for models trained on transformed data. 
Given a trained machine learning model and its confidence score on a datapoint, MIA infers whether the datapoint was part of the model's training data. 
In order to conduct MIA, an attacker trains a {\it shadow model} that he/she expects to mimic the {\it target model} under attack.
Once trained, the {\it shadow model's} confidence scores on the datapoints {\it members} of its training set and other { \it non-member} datapoints are used to train the binary attack model.
Given a datapoint, the attacker then extracts a similar vector of confidence scores from the {\it target model} and uses the {\it attack model} to make a {\it member}/{\it non-member} prediction. 

\section{ADePT: Auto-encoder based DiffEerentially Private Text transformation}
Consider an utterance $\mathbf{u}$ drawn from a dataset $\mathcal{D}$.
Furthermore, consider an \stos model that takes input a sentence $\mathbf{u}$ and outputs another sentence $\mathbf{v}$. 
A vanilla \stos model consists of an encoder that returns a vector representation $\mathbf{r} = \textsf{Enc}({\mathbf{u}})$ for the input $\mathbf{u}$, which is then passed onto the decoder that constructs an output $\mathbf{v} = \textsf{Dec}(\mathbf{r})$. 
We define ADePT as a randomized algorithm $\mathcal{A}$, that given an utterance $\mathbf{u}$, generates $\mathbf{v}$ as shown in equation \ref{eq:transform}. 
$\mathbf \eta$ is a vector sampled from either a Laplacian or a Gaussian distribution (with 0 mean and a pre-defined variance).
\vspace{-2mm}

\begin{equation}\label{eq:transform}
 \mathbf{v} = \textsf{Dec}(\mathbf{r}^\prime)
\end{equation}

\vspace{-2mm}
\begin{equation}\label{eq:transform}
\text{Where}\;\; \mathbf{r}^\prime = \textsf{Enc}(\mathbf{u})\cdot\min\Big({1,\frac{C}{||\textsf{Enc}(\mathbf{u})||_2}}\Big)+ \mathbf{\eta}
\end{equation}

\subsection{Proof that ADePT is differentially private}
Given that ADePT conducts a transformation from $\mathbf{u} \xrightarrow{} \mathbf{r}^\prime \xrightarrow{} \mathbf{v}$, we first show that it is sufficient to prove that the transformation from $\mathbf{u} \xrightarrow{} \mathbf{r}^\prime$ is DP for ADePT to be DP. 
Thereafter, we prove that the transformation $\mathbf{u} \xrightarrow{} \mathbf{r}^\prime$ is DP. 

\begin{lemma}
The transformation $\mathbf{u}  \xrightarrow{} \mathbf{v}$ will be at least ($\epsilon, \delta$) differentially private, if the algorithm that transforms $\mathbf{u}$ to $\mathbf{r}^\prime$ is ($\epsilon, \delta$) DP.
\end{lemma}

\begin{proof}
This is true based on proposition 2.1 on post-processing in 
\citet{dwork2014algorithmic}. 
\end{proof}

\begin{theorem}
If $\mathbf{\eta}$ is a multidimensional noise, such that each element $\eta_i \in \mathbf{\eta}$ is independently drawn from a distribution shown in equation \ref{eq:lap}, then the transformation from $\mathbf{u} \xrightarrow[]{} \mathbf{v}^\prime$ is $(\epsilon, 0)$ DP. 
\end{theorem}

\begin{equation}\label{eq:lap}
 Lap(\eta_i) \sim \frac{\epsilon}{4C}\exp (-\frac{\epsilon|v_i|}{2C})
\end{equation}

\begin{proof}
We refer the reader to the proof in \citet{dwork2014algorithmic}, Theorem 3.6. 
The function $f(x)$ used in the Theorem in \citet{dwork2014algorithmic} is equivalent to the encoder output with clipping. 
The l1-sensitivity of this function (please refer to definition 3.1 in \citet{dwork2014algorithmic}) is $2C$ since maximum L1 norm difference between two points in a hyper-sphere of radius $C$ is $2C$. 
Replacing $\Delta f$ in Theorem 3.6 in \citet{dwork2014algorithmic} by $2C$, we obtain the that the transformation is $(\epsilon, 0)$ DP.
\end{proof}

Akin to the proof with a Laplacian noise, we can also borrow the proof in Appendix A in \citet{dwork2014algorithmic} to show that ADePT would also be DP if $\mathbf{\eta}$ was a Gaussian noise.

\section{Experimental setup}

We perform an intent classification task in our experiments and quantify impacts on accuracy and privacy metrics after data transformation via the ADePT mechanism.
While the intent classification accuracy quantifies the utility of the transformed dataset, we evaluate success of MIA against the IC model to quantify privacy.
Below, we describe the datasets, auto-encoder and IC model training and the MIA setup used in our experiments. 

\begin{figure*}
\centering
\begin{subfigure}{0.5\textwidth}
\vspace{-6mm}
  \centering
  \includegraphics[width=\linewidth]{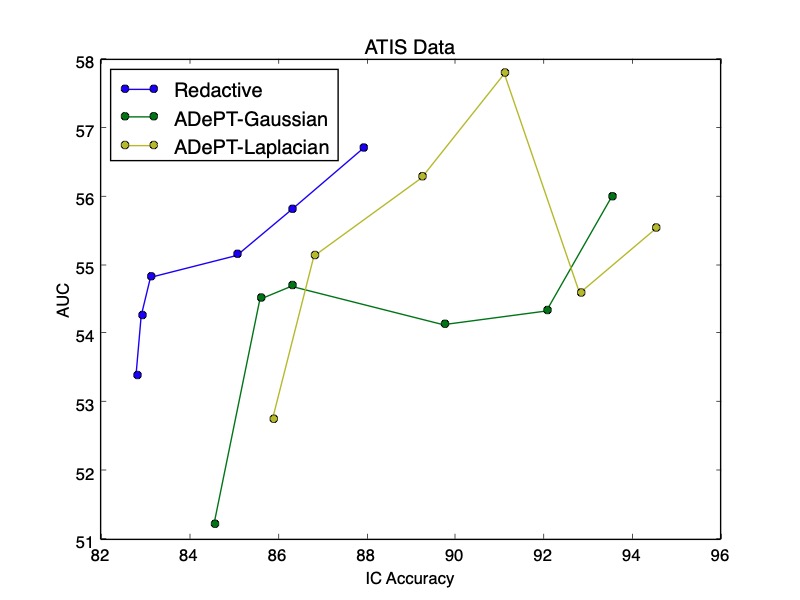}
\end{subfigure}%
\begin{subfigure}{0.5\textwidth}
\vspace{-6mm}
  \centering
  \includegraphics[width=\linewidth]{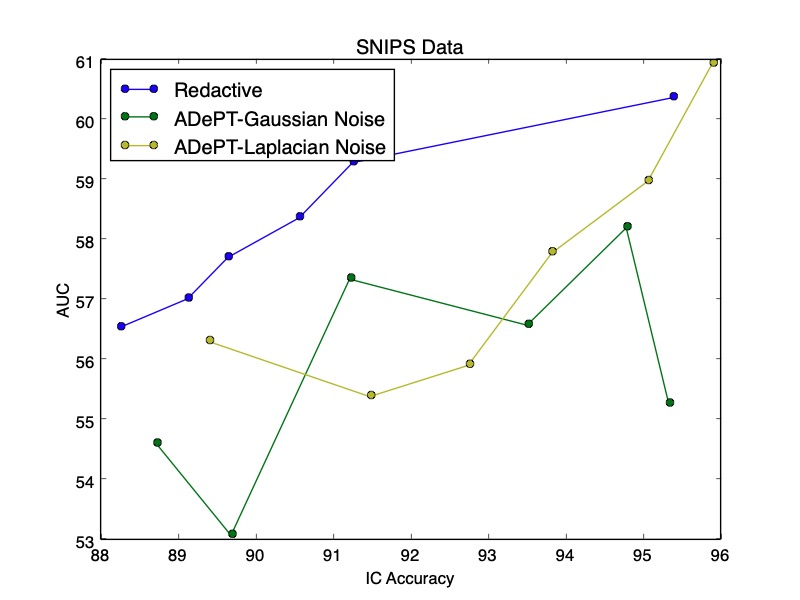}
\end{subfigure}
\vspace{-4mm}
\caption{Privacy and accuracy metrics using baseline and ADePT mechanisms on the ATIS and SNIPS datasets. Baseline mechanism transforms datasets with Laplacian noise with variance values $\in (1,6,9, 15, 28, 100)$. ADePT transforms datasets with Gaussian and Laplacian noises with variances $\in (0.25, 0.5, 0.6, 0.75, 0.85, 1)$. The variance scales are different between the two mechanism due to inherent difference in their construct.}
\label{fig:test}
\end{figure*}

\begin{table*} [!htb]
\centering
\caption{Example of a good and a corrupted output from ADePT}
\vspace{-4mm}
\begin{tabular}{lrl} 
\hline \multicolumn{1}{|m{5cm}|}{\textbf{Original}} & \multicolumn{1}{>{\centering}m{43mm}|}{\textbf{Baseline}} & \multicolumn{1}{>{\centering}m{43mm}|}{\textbf{ADePT}} \\ 
\hline \multicolumn{1}{|m{5cm}|}{what are the flights on \textbf{january first 1992 from boston to san francisco}} & \multicolumn{1}{m{5cm}|}{what are the flights on \textbf{february inhales 1923 from boston to san mostrar}} & \multicolumn{1}{m{5cm}|}{what are the flights on \textbf{thursday going from dallas to san francisco}} \\  \hline \multicolumn{1}{|m{5cm}|}{show all \textbf{ flights boston to any time}} & \multicolumn{1}{m{5cm}|}{show all \textbf{5-minutes distinctions from massachusetts to tempat chiefs}} & \multicolumn{1}{m{5cm}|}{show all \textbf{flights flights flights boston to any time}} \\ \hline
\end{tabular}
\vspace{-4mm}
\label{tab:examples}
\end{table*}

\subsection{Datasets}
We use ATIS \cite{dahl1994} and SNIPS \cite{coucke2018} for training IC models on the respective datasets.
The ATIS dataset consists of $\sim$5.5k data samples, while the SNIPS dataset consists of $\sim$14.5k data samples.
We used a 50:50 split for training and evaluation sets.
Apart from offering a larger accuracy evaluate test set, a 50:50 split also ensures that we have a balanced training and evaluation sets for MIA, as discussed in Section~\ref{sec:mia}. 

\subsection{Training the auto-encoder model}
Given utterances $\mathbf{u}$ in the training partition of the datasets of interest, we train an \stos model to reconstruct the input utterance $\mathbf{u}$ via the decoder $\textsf{Dec}$.
In our case, the auto-encoder is a sequence to sequence model, where both encoder and decoders are uni-directional LSTM models.
We train the auto-encoder on the training portions of the ATIS and SNIPS datasets, with an objective to reconstruct the input sentence through the {\it latent} representation.
Note that during training, we apply clipping to ensure that the latent representation are encouraged to reside within a hyper-sphere of radius $C$, no noise is added to the latent representation.
Clipping and noising operations are applied during the final transformation after the auto-encoder is trained, as discussed in section~\ref{sec:transformation}. 

\subsubsection{Making ADePT utility preserving}
In the proof above, we show that ADePT is DP algorithm that transforms input utterances $\mathbf{u}$ to $\mathbf{v}$.
For the purposes of training an intent classifier, a naive scheme can assume that the intent label applied to the utterance $\mathbf{u}$ is also applicable to $\mathbf{v}$.
However, this assumption may not always be true as the transformation may render utterance $\mathbf{v}$ to carry a different intent label than $\mathbf{u}$. 
In order to encourage the transformed utterances $\mathbf{v}$ to conform to the intent label for utterance $\mathbf{u}$, and also obtain the correct intent label in cases where the transformation may lead $\mathbf{v}$ to belong to a different intent, we tweak the \stos model to also ingest the intent label.
We train annotation aware \stos models with inputs/outputs as utterances and the corresponding intent.
The intent label is appended to the beginning of each utterance (demarcated with a special character to help distinguish the intent names with utterance tokens) during the auto-encoder training.

\subsection{Data transformation}
\label{sec:transformation}
Once the \stos model is trained, we apply the transformation again on the training portions of ATIS and SNIPS datasets.
During the transformation, the intent token is appended with the rest of the utterance and an output in a similar format is expected. 

\subsection{Intent classifier training}
The ADePT transformation yields the altered sentence, along with an intent.
We transform the training portion of ATIS and SNIPS datasets through the autoencoder and use the altered sentences along with the reproduced intent for training an intent classifier.
Our IC architecture is inspired from \citet{CLC2016} and consists of three blocks: (i) an embedding block consisting of word and character embeddings, (ii) a block consisting of bi-directional LSTM layers and, (iii) a fully connected network operation on a max-pool of LSTM layer outputs for intent classification.

\subsection{Privacy evaluation using MIA}
\vspace{-1mm}
\label{sec:mia}

We train the attack model on confidence scores returned by a shadow IC model trained similarly as the target IC model. 
We extract scores for the top five intents returned by the shadow IC model on the {\it member} and {\it non-member} sentences used to train the shadow IC model.
The attack model is a binary logistic regression model, trained on the extracted IC scores from `member' and 'non-member' sentences.

During the attack, top 5 intent scores from the target IC model are fed to the logistic regression model to make a prediction whether the corresponding scores belong to the target model's {\it member} or {\it non-member} data. 
While the {\it member} sentences are sourced from the training set, we borrow {\it non-member} sentences from the test set used to evaluate the model accuracy (note that their counts are balanced as we use a 50:50 split).
We use the Area Under the ROC curve (AUC) to evaluate the success of the attack model and a higher AUC implies worse privacy metric.

\vspace{-1mm}

\section{Experimental Results}
We conduct ADePT transformation using both Laplacian and Gaussian noises, with different variance values.
The baseline mechanism also uses a Laplacian noise to sample words replacements for the private words. 
Figure~\ref{fig:test} show the MIA success rates and IC accuracies obtained on ATIS and SNIPS data respectively.
Note that the algorithm with a lower AUC and a higher IC accuracy is desirable.
We observe that as we sweep the noise parameters for the ADePT and Redactive mechanisms, we generally obtain lower AUC with a higher IC accuracy for the former. Additionally ADePT mechanism with a Gaussian noise performs the best.
This empirical observation supports our hypothesis that factoring in the intent label during ADePT based transformation helps providing better utility.  

However, we also note that the privacy-utility trade-off in ADePT can be non-monotonic.
We noticed that the sentence transformation using encoders is sensitive to noise value added to encoded representation $ \textsf{Enc}(\mathbf{u})$. The clipping and noise additional has potential to change the entire sentence, unlike the baseline, where the public phrase in the utterance remains unaltered and only the private phrases in the utterances are subject to alteration.
We show two examples of sentence transformation using the baseline and Gaussian ADePT mechanism in Table ~\ref{tab:examples}.
In particular, the decoder tends to repeat the same word multiple times for corrupted outputs which can be corrected with constrained decoding.

\section{Conclusions}
We propose ADePT - an \stos based DP algorithm in this paper.
We theoretically prove that the mechanism is DP and demonstrate that it offers a better privacy utility trade-off compared to a baseline that relies on detecting the transforming public phrases in a sentence.
In the future, we will extend ADePT to transforming datasets with sequence level tags (for instance, in named entity recognition tasks) and also use non-autoregressive decoders (e.g. transformers).
We will also extend the mechanism to other modalities (e.g. Image) using auto-encoder models in the corresponding domains. 

\bibliography{anthology,eacl2021}
\bibliographystyle{acl_natbib}

\end{document}